\documentclass[Royal,sageh,times]{sagej}

\usepackage{moreverb,url}
\usepackage{subfigure}
\usepackage{afterpage}
\newtheorem{lemma}{Lemma}
\newtheorem{remark}{Remark}
\newtheorem{theorem}{Theorem}
\newtheorem{assumption}{Assumption}

\usepackage[colorlinks,bookmarksopen,bookmarksnumbered,citecolor=red,urlcolor=red]{hyperref}
\usepackage{tikz,xcolor,hyperref}
\definecolor{lime}{HTML}{A6CE39}
\DeclareRobustCommand{\orcidicon}{%
	\begin{tikzpicture}
		\draw[lime, fill=lime] (0,0) 
		circle [radius=0.16] 
		node[white] {{\fontfamily{qag}\selectfont \tiny ID}}; 
		\draw[white, fill=white] (-0.0625,0.095) 
		circle [radius=0.007];	  
	\end{tikzpicture}
	\hspace{-2mm}}
\foreach \x in {A, ..., Z}{%
	\expandafter\xdef\csname orcid\x\endcsname{\noexpand\href{https://orcid.org/\csname orcidauthor\x\endcsname}{\noexpand\orcidicon}}
}
\makeatother

\newcommand\BibTeX{{\rmfamily B\kern-.05em \textsc{i\kern-.025em b}\kern-.08em
T\kern-.1667em\lower.7ex\hbox{E}\kern-.125emX}}

\begin{document}


\title{Distributed Secondary Frequency Control and SoC Balance for Droop-Controlled BESSs with A Unified SoC Relative Variation Rate}
\author{Yalin Zhang\affilnum{1,2}\orcidA{},
	Zhongxin Liu\affilnum{1,2}\orcidB{},
	Fuyong Wang\affilnum{1,2}\orcidC{},
	and Zengqiang Chen\affilnum{1,2}\orcidD{}}

\affiliation{\affilnum{1}College of Artificial Intelligence, Nankai University, Tianjin 300350\\
\affilnum{2}Tianjin Key Laboratory of Interventional Brain-Computer Interface and Intelligent Rehabilitation, Nankai University, Tianjin 300350, China}

\corrauth{Fuyong Wang, College of Artificial Intelligence, and Tianjin Key Laboratory of Interventional Brain-Computer Interface and Intelligent Rehabilitation, Nankai University, Tianjin 300350, China}

\email{wangfy@nankai.edu.cn}

\begin{abstract}
	The State of Charge (SoC) balance, power sharing, and frequency restoration are common control objectives of battery energy storage systems (BESSs). However, the SoC balance scheme induced by the power allocation through existing droop controllers can cause the capacity parameters of battery cells to be unequal to the droop coefficient, which is the result of battery capacity degradation. Under this limitation, previous capacity based droop controllers and the secondary controllers no longer suitable to address the imprecise power sharing and frequency restoration caused by this problem. Therefore, a power allocation scheme based on the current SoC level ratio is designed to induce a new droop controller and ensure that the SoC simultaneously drops to 0. In order to restore frequency in a distributed manner and obtain SoC level ratio, a distributed nominal frequency controller, SoC average estimator, and power sharing controller are designed based on multi-agent systems (MASs) in both asymptotic and finite time manners. In the asymptotic scheme, the steady-state performance of the SoC estimator is adjustable, and power sharing and frequency restoration are zero errors. In the finite time scheme, the influence of parameters on convergence time is well analyzed, and some conservative calculation methods are provided. Several time-domain simulation examples are
	designed on an improved IEEE57 bus system to verify the distribution of asymptotic and finite time schemes.
\end{abstract}


\maketitle
\begin{center}
	N\footnotesize{OMENCLATURE}
\end{center}
\normalsize
\begin{tabbing}
	\hspace{2cm} \= \kill
	BESS\> battery energy storage system\\
	SoC\>  the State-of-Charge\\
	MAS\> multi-agent system\\
	power\> proportional output power\\
	$E_i$, $K_i^E$, $K_i^P$\> the SoC, capacity, and droop coefficient\\
	$\omega_i$, $\omega_i^*$\> the frequency and nominal frequency\\
	$u_i^P$, $u_i^\omega$\>the control inputs\\
	$K_{i,m}^E$, $\eta_i$, $f_i$\>the rated capacity, Coulomb efficiency, and\\ \>current cycle number\\
	$P_a$, $E_a$\> the average power and SoC\\
	${\hat P}_{a,i}$, ${\hat E}_{a,i}$\> the estimated average power and SoC\\
	$\tilde P_i$, $P_i$\> output power and its proportional value\\
	$\mathrm{K}_i^E$\> a constant defined as $K_i^E=\frac{1}{Q_iV_i}$\\
	$\hat P_{a,i}$, $\hat E_{a,i}$\> the estimated average power and SoC \\
	$\epsilon_E$\> the control accuracy\\
	$\cal G$, $\cal V$\> the communication graph and vertex set\\
	$\cal E$, $\cal L$\> the edge set and Laplacian matrix\\
	$a_{ij}$, $d_i$\> the communication weight from $j$ to $i$,\\ 
	\>the in-degree of $i$\\
	$\cal A$, ${\cal D}$\> the adjacency matrix and in-degree matrix\\
	$\cal B$, $b_{ii}$\> the adjacency matrix between the leader and\\
	\>all follower and its component\\
	$q_i$, $v_i$\> the intermediate states\\
	$\alpha$, $\beta$, $\beta_1$, $\beta_2$, $\kappa$\> the control gains\\
	$T_P$, $T_E$, $T_\omega$\> the controller stability time\\
	$\omega^r$\> the frequency of the main grid\\
\end{tabbing}
\section{Introduction}
Battery energy storage system (BESS) is integrated into the smart grid to suppress the peak and valley supply gap and suppress the randomness and intermittency of renewable energy generation \cite{Lin2022, Pei2022}, thereby improving power quality and providing reliable uninterrupted power supply \cite{Pei2022}. However, the management of multiple parallel BESSs also faces challenges \cite{rouholaminiReviewModelingManagement2022,caleroReviewModelingApplications2022,Kang2022}. The imbalance of State of Charge (SoC), serving as an important indicator to measure the level of a BESS, leads to the waste of battery capacity  \cite{caleroReviewModelingApplications2022,Kang2022,Xu2022,Kebede2022,Qays2022}, which increases management costs. The differential restoration of frequency leads to power harmonics, thereby reducing power quality \cite{Pei2022,rouholaminiReviewModelingManagement2022,Kang2022,Xu2022,hossainlipuReviewControllersOptimizations2022}. Hence, cooperative management/control of multiple parallel BESSs becomes a problem worthy of in-depth study.
\par Thus, a large number of excellent literature has been committed to solving the above problems. Both distributed and centralized manners have been investigated. For example, SoC of each battery was monitored by a centralized controller in \cite{Gonzalez-Garrido2020}, while external balancing circuit is applied to realize power sharing. However, the external balancing circuit was going to consume some energy. In practice, centralized controller is more expensive, and the single-point failures are unavoidable \cite{rouholaminiReviewModelingManagement2022,caleroReviewModelingApplications2022,forero-quinteroProfitabilityAnalysisDemandside2022}. Instead, distributed control represented by multi-agent systems (MASs) can enhance the robustness and scalability of the system, making it more favored for managing BESSs.
\par For multiple BESSs, many distributed algorithms based on MAS for BESS focus on asymptotic control \cite{8076899, Nguyen2021a,Zeng2022,Raeispour2020a}, finite time control \cite{Ding2020, Zhang2020}, and disturbance rejection/resilience control \cite{Ding2020,Raeispour2020a}, etc, of SoC balancing, power sharing and frequency restoration. Therefore, in terms of convergence performance \cite{8076899, Nguyen2021a,Zeng2022,Raeispour2020a,Ding2020, Zhang2020}, robustness performance \cite{Ding2020,Raeispour2020a}, etc., these results perform extremely well. The model-free method is also used to promote SoC consensus and frequency regulation based on droop control in \cite{Chen2022a}. The author improved the power allocation scheme based on capacity using an adaptive term based on SoC in \cite{Yang2021d} and achieved SoC consensus. In \cite{Yu2021}, the author allocates power based on marginal cost to achieve optimal economic cost operation. These distributed algorithms for BESS are hierarchical and have some commonalities. Therein, power needs to be proportionally shared according to the capacity of each BESS, resulting in SoC consensus and only power participation in frequency modulation based on droop control. 
\par However, the consensus method of SoC based on power exchange has changed the working modes of some BESSs, that is, from discharge mode to charging mode. This not only leads to an increase in losses on the power feeder line, but also causes some existing battery parameters to change due to early entry into the next cycle, thus requiring a re-estimation of the parameters. In other words, the current methods have caused waste of power and capacity, as well as modeling errors.
\par Recently, some new developments regarding SoC balance have emerged in \cite{Xing2019,Meng2021,Meng2022, Wu2022}. The results of these studies are all aimed at verifying that when BESSs have a uniform discharge rate, no battery will exit prematurely due to power depletion, in other words, all batteries will run out of energy at the same time. There is a power allocation scheme implied here. Unlike previous balance schemes based on SoC consensus, these results seem to indicate that as long as power is shared according to the current SoC, the goal of achieving time consensus for SoC depletion to 0 can be achieved. Based on this, the previous droop control as a primary control strategy cannot achieve accurate frequency restoration. Thus, inspired by the above related work, considering the novel power allocation scheme, we replace the original static value with a dynamic droop coefficient based on SoC. This dynamic value includes the rated capacity of the battery, real-time SoC, and its average value. In order to obtain this average value and achieve secondary frequency control in distributed manners, we are going to design two types of distributed estimators for the SoC balance and two active power sharing controllers, which are later referred to as power sharing estimators, for multi-BESSs network with the heterogeneous dynamics in this paper, respectively. That is, asymptotic and finite time schemes are simultaneously designed here. Concisely, the main contributions in this paper are listed as follows.
\begin{enumerate}	
	\item[(i)] This article introduces a discharge rate constraint for each battery, expecting each BESS to have the same relative discharge rate. The constraint of the discharge mode introduced in this article ensures that all batteries are depleted simultaneously and that the operating mode of each one will be unchanged before depletion. Compared to existing schemes in \cite{9889208, 8076899, Nguyen2021a,Zeng2022,Raeispour2020a,Ding2020, Zhang2020}, our designed scheme does not require power exchange, thereby reducing losses on power feeders.
	\item[(ii)] A new frequency droop controller is proposed. Active power is no longer allocated based solely on the rated capacity of each battery, but rather on the rated capacity and the current SoC level of each BESS and the average SoC and output power states in this network. That is to say, the frequency droop coefficient is dynamic and non constant, which is different from before \cite{8076899, Nguyen2021a,Zeng2022,Raeispour2020a,Ding2020, Zhang2020}. The advantage of doing so is that the power allocated by this droop control can ensure that all BESSs are in discharge mode at all times and depleted simultaneously.
	\item[(iii)] In order to obtain the average SoC and output power states and achieve secondary frequency control and SoC balance, two distributed schemes are proposed based on an improved droop control strategy, one is asymptotic and the other is finite time.
\end{enumerate}
\par Section \textcolor{yellow}{2} gives a general description of graph theory. Section \textcolor{yellow}{3} introduces a common model of BESS with droop control and the common control objectives. \textcolor{yellow}{Sections 4} and \textcolor{yellow}{5} present asymptotic and finite-time schemes, respectively. Section \textcolor{yellow}{6} designs several cases to test the proposed estimators. Section \textcolor{yellow}{7} summarizes this paper.
\section{The Description of \textcolor{blue}{A} BESS and Problem Statement} \label{II}
\subsection{\textcolor{blue}{A} Frequency droop \textcolor{blue}{controlled} BESS}
\par \textcolor{blue}{A} BESS usually adopts \textcolor{blue}{a} hierarchical control scheme. In this paper, we adopt a multi-BESS network, and the simplified dynamics of each one can be described in \textcolor{blue}{the} \textit{dq} frame as (\ref{1a})-(\ref{1d}), referred to \cite{8076899, Nguyen2021a, Raeispour2020a, Ding2020, Zhang2020} for detail.
\begin{subequations}
	\begin{equation}
		\label{1a}
		\dot E_i=\frac{- K_i^E}{3600}{\tilde P}_i + u_i^E,
	\end{equation}
	\begin{equation}
		\label{1b}
		\dot {\tilde P}_i = u_i^P,
	\end{equation}
	\begin{equation}
		\label{1c}
		\omega_i=\omega_i^* - K_i^P{\tilde P}_i,
	\end{equation}
	\begin{equation}
		\label{1d}
		\dot \omega _i^* = u_i^\omega,
	\end{equation}
\end{subequations} 
where ${E_i}$ is \textcolor{blue}{the} SoC of battery $i$, ${P_i}$ and ${\omega _i}$ are the output active power and \textcolor{blue}{the} frequency, respectively, of inverter $i$, $\omega _i^*$ is represented for the nominal value of the frequency, \eqref{1c} is droop control, namely the primary control for frequency, $K_i^P$ is \textcolor{blue}{a} droop coefficient. Therefore, BESS models conforming to these dynamics are heterogeneous. Previous researchers have favored using $K_i^P =  {K_i^E}$ and denoting ${P_i} \buildrel \Delta \over = K_i^P{\tilde P_i}$, which means that active power is allocated according to the capacity of battery $i$.
\par In previous SoC balance schemes, both underactuated and fully actuated schemes are actually induced by output power to achieve SoC consensus. Driven by these
schemes, some batteries at low SoC levels are in charging mode for a period of time, while others are in discharge mode. The inconsistency of working modes has led to batteries in charging mode entering the next cycle ahead of schedule. However, we note that the capacity of the battery decreases with an increase in the number of cycles. Specifically, for BESS $i$ in the $f_i$-th cycle, its
current actual capacity is
$$K_i^E=K_{i,m}^E\eta_i^{f_i},$$
where $K_{i,m}^E$ and $\eta_i$ are the rated capacity and the Coulomb efficiency of BESS $i$, respectively. Therefore, treating $K_i^E$ as a constant may cause significant modeling errors.
\subsection{A novel power allocation strategy and its induced droop controller}
\par The significance of SoC balance is to ensure that the battery cells within the battery pack are simultaneously depleted, thereby fully utilizing the capacity of the battery pack. The previous schemes of directly aligning SoC and proportionally aligning output power both require power exchange between battery packs, resulting in unnecessary waste on power feeders. To alleviate this situation, for a BESS network containing $n$ batteries, the output power of each battery often needs to be determined based on its SoC. Mathematically, batteries have the same relative discharge rate i.e.,
\begin{equation}
	\label{2}
	\frac{P_i}{E_i}=\frac{P_j}{E_j},\forall i,j \in \left\{ 1,2, \cdots ,n \right\}\textcolor{blue}{.}
\end{equation}
\textcolor{blue}{Combining} (\ref{1a}) and (\ref{2}), it can be concluded that
the output power of each BESS can be counted as
\begin{equation}
	\label{3}
	P_i=\frac{E_i}{E_a}{P_a},
\end{equation}
where $E_a$ and $P_a$ are the average value of \textcolor{blue}{SoCs and} the total load power\textcolor{blue}{, respectively}. As long as the power is allocated according to this scheme, SoC balance can be guaranteed without the need for additional control inputs.
\par Based on (\ref{2}), we can conclude that active power needs to be allocated according to SoC level. If $K_i^P=K_i^E\frac{E_a}{E_i}$, then $P_a=K_i^P\tilde{P}_i$. In this way, as a primary controller, the droop control can be modified as follows,
\begin{equation}
	\label{4}
	\omega_i=\omega_i^*-P_a.
\end{equation}
\begin{remark}
	For SoC balance schemes that rely on a common frequency droop control shown in \eqref{1c} to allocate power, $K_i^P=K_i^E$ is constant and a widely used condition, which is very unreasonable from the above analysis. And our redesigned droop control \eqref{4} adopts a dynamic droop control gain, to allocating power \textcolor{blue}{according to the current SoC level} to ensure that all BESSs are in discharge mode and depleted simultaneously.
\end{remark}
\subsection{Control objectives}
\par Frequency restoration helps to provide high-quality electricity. Before that, SoC balance and power sharing states should be well distributed estimated. Hence, each BESS is subjected to the follow dynamics,
\begin{subequations}
	\label{5}
	\begin{equation}
		\label{5a}
		\dot E_i=\frac{-1}{3600}P_i,
	\end{equation}
	\begin{equation}
		\label{5b}
		\dot{\hat E}_{a,i}=u_i^E,
	\end{equation}
	\begin{equation}
		\label{5c}
		P_i=\frac{E_i}{\hat E_{a,i}}\hat P_{a,i},
	\end{equation}
	\begin{equation}
		\label{5d}
		\dot{\hat P}_{a,i}=u_i^P,
	\end{equation}
	\begin{equation}
		\label{5e}
		\omega_i=\omega_i^*-{\hat P}_{a,i},
	\end{equation}
	\begin{equation}
		\label{5f}
		\dot \omega _i^* = u_i^\omega,
	\end{equation}
\end{subequations} 
where \eqref{5c} and \eqref{5e} are \textcolor{blue}{a} power allocation scheme and a modified droop controller, respectively, $\hat E_{a,i}$ and $\hat P_{a,i}$ are the estimation of $E_a$ and $P_a$\textcolor{blue}{, respectively}. The detailed objectives are described in Problem 1,
\newtheorem{problem}{Problem}
\begin{problem} \label{problem}
	Consider a multi-BESS network including $N$ BESSs with droop control and subjected to dynamics (\ref{1a})-(\ref{1d}). We intend to design \textcolor{blue}{some} distributed \textcolor{blue}{controllers} for each BESS to solve the following problems, 	
	\\(1) The real time active power sharing state of each battery can be well estimated, and reactive sharing can be achieved, i.e.,
	\begin{center}
		$\mathop {\lim }\limits_{t \to \infty }\left| {{\hat P_{a,i} } - {P_a}} \right| = 0$, $\forall i \in \left\{ {1,2, \cdots ,N} \right\}.$
	\end{center}
	(2) On the premise that $\hat P_{a,i}$ is well distributed estimated, a nominal frequency controller is designed to achieve secondary control, i.e.,
	\begin{center}
		$\mathop {\lim }\limits_{t \to \infty }\left| {{\omega _i} - {\omega ^r}} \right| = 0$,
		$\forall i \in \left\{ {1,2, \cdots ,N} \right\},$
	\end{center}
	where $\omega^r$ denotes the reference frequency.\\
	(3) Thus, to achieve SoC balance, the real time SoC balance state should be well estimated, i.e., 
	\begin{center}
		$\mathop {\lim }\limits_{t \to \infty }\left| {{\hat E_{a,i}} - E_a} \right| \le \epsilon$, $\forall i \in \left\{ {1,2, \cdots ,N} \right\},$\\
	\end{center}
	where $\epsilon$ is a prescribed positive constant.
\end{problem}
\section{Distributed secondary control schemes design}
\subsection{Graph theory}
\par In general, agents communicate with each other on a topology ${\cal G} = ({\cal V},{\cal E})$. ${\cal V}$ stands for the set embracing $n$ nodes, each of which represents an agent. Meanwhile, ${\cal E}$ is defined as the set of edge, if there exists an edge from nodes $j$ to $i$, ${a_{ij}} =1$, which means the agent $i$ has access to the agent $j$'s information; if not, ${a_{ij}} = 0$. Besides, self-edge is not considered in this paper, i.e., ${a_{ii}} = 0$ for all nodes. Such a graph is called the directed graph, which contains a spanning tree and a root node that has no parent. According to the above statement, we define ${\cal A} = [{a_{ij}}]$ as the adjacency matrix, and ${\cal D} = diag({d_1},{d_2}, \cdots ,{d_n})$ as the in-degree matrix with ${d_i} = \sum\limits_{j \ne i} {{a_{ij}}} $. Therefore, the Laplacian matrix corresponding ${\cal G}$ is ${\cal L} = {\cal D} - {\cal A}$.
\par The Leader-Follower model is mainly applied in this paper. Let ${\cal B} = diag\{ {b_{11}},{b_{22}}, \cdots ,{b_{nn}}\} $ be \textcolor{blue}{an} adjacency matrix between the leader and all \textcolor{blue}{followers}. If the information of the leader can be acquired by the follower $i$, ${b_{ii}} = 1$; otherwise, ${b_{ii}} = 0$. Denote ${\cal H}={\cal L}+{\cal B}$.
\begin{assumption}
	\label{assum1}
	Assume that $\cal G$ is a connected undirected graph, and at least one follower can access the leader.
\end{assumption}
\begin{lemma}\label{lemma3}
	\cite{Mesbahi2010a} Under Assumption \ref{assum1}, $\cal L$ is positively semidefinite, and ${\cal H}={\cal L}+{\cal B}$ is positively definite.
\end{lemma}
\par Under Assumption \ref{assum1}, $\cal L$ has a zero eigenvalue and $n-1$ positive eigenvalues, and $\cal H$ has $n$ positive eigenvalues. Let $0=\lambda_1({\cal L})<\lambda_2({\cal L})\le\cdots\le\lambda_n({\cal L})$ and $\lambda_m=\lambda_1({\cal H})\le\lambda_2({\cal H})\le\cdots\le\lambda_n({\cal H})=\lambda_M$ be the non decreasing eigenvalue sequences of $\cal L$ and $\cal H$\textcolor{blue}{, respectively}.
\subsection{Distributed asymptotical control scheme}
Inspired by \cite{Wu2022, George2019, Meng2021}, the distributed SoC balance and power sharing estimators can be designed for each battery unit here, as shown in (\ref{6a})-(\ref{6c}),
\begin{subequations}
	\label{6}
	\begin{equation}
		\label{6a}
		\hat E_{a,i}=q_i+E_i,
	\end{equation}
	\begin{equation}
		\label{6b}
		\dot q_i=-\alpha q_i-\beta\sum\limits_{j \in N_i} a_{ij}(\hat E_{a,i}-\hat E_{a,j})-v_i,
	\end{equation}
	\begin{equation}
		\label{6c}
		\dot v_i=\alpha\beta\sum\limits_{j \in N_i} a_{ij}(\hat E_{a,i}-\hat E_{a,j}),
	\end{equation}
	\begin{equation}
		\label{6d}
		\dot{\hat P}_{a,i}=-\kappa\sum\limits_{j \in N_i} a_{ij}({\hat P}_{a,i}-{\hat P}_{a,j}),
	\end{equation}
\end{subequations}
where ${v_i}$ and $q_i$ are the designed intermediate states. Here, let initial condition ${\hat E_{a,i}} = {E_i}$, ${v_i} = 0$, $q_i=0$, and ${\hat P_{a,i}} = {P_i}$ for $i \in \left\{ {1,2, \cdots ,N} \right\}$, $\alpha >0$, $\beta  > 0$ and $\kappa  > 0$ are design parameters. In addition, the nominal frequency controller can be constructed in (\ref{7}),
\begin{equation}
	\label{7}
	\dot\omega_i^*=-\kappa(\sum\limits_{j = 1}^N a_{ij}(\omega_i^*-\omega_j^*)+b_{ii}(\omega_i^*-{\hat P}_{a,i}-\omega ^r))\textcolor{blue}{.}
\end{equation}
Due to the fact that the nominal frequency is not \textcolor{blue}{a variable that needs to be measured} of a BESS, the use of this controller can reduce the installation of measurement equipment and measurement errors.
\par To illustrate the stability analysis of (\ref{6}) and (\ref{7}), the following useful lemmas are given.
\begin{lemma} \label{lemma1}
	\cite{George2019} Consider the system dynamics as
	$${{{\dot x}_i} = {{\dot u}_i} - \alpha \left( {{x_i} - {u_i}} \right) - \beta \mathop \sum \limits_{j = 1}^N {L_{ij}}{x_j} - {v_i}}\textcolor{blue}{,}$$	
	$${{{\dot v}_i} = \alpha \beta \mathop \sum \limits_{j = 1}^N {L_{ij}}{x_j}}\textcolor{blue}{.}$$
	Let
	$${{y_i} = {x_i} - \frac{1}{N}\mathop \sum \limits_{j = 1}^N {u_j},\quad i \in \{ 1, \ldots ,N\} }\textcolor{blue}{,}$$
	$${w = v - \bar v,\quad \bar v = {\Pi _N}(\dot u + \alpha u)}\textcolor{blue}{,}$$
	where $u=[u_1,u_2, \cdots, u_N]$, ${\Pi _N} = {I_N} - \frac{1}{N}{1_N}1_N^T$. Then, the system can be rewritten as follow
	$$\left[ {\begin{array}{*{20}{c}}
			{\dot y}\\
			{\dot w}
	\end{array}} \right] = A\left[ {\begin{array}{*{20}{l}}
			y\\
			w
	\end{array}} \right]\textcolor{blue}{,}$$
	where
	$$A = \left[ {\begin{array}{*{20}{c}}
			{ - \alpha {I_N} - \beta L}&{ - {I_N}}\\
			{\alpha \beta L}&0
	\end{array}} \right]\textcolor{blue}{.}$$
	If $\alpha $, $\beta  > 0$, and the communication topology satisfies Assumption 1, then
	${y_i}(t) \to  - \frac{{{\alpha ^{ - 1}}}}{N}\mathop \sum \limits_{j = 1}^N {w_j}(0)$,
	${w_i}(t) \to \frac{1}{N}\mathop \sum \limits_{j = 1}^N {w_j}(0)$, as $t \to \infty $, $\forall i \in \{ 1, \ldots ,N\} $\textcolor{blue}{.}
\end{lemma}

\begin{lemma}\label{lemma2} 
	\cite{Kia2014a} Based on Lemma \ref{lemma1}, there exist a positive constant $\gamma_1$, if $\alpha>0$ and $\beta>0$, $y_i$ satisfies  
	$$\mathop {\lim }\limits_{t \to \infty } \left| {y_i} \right| \leq \frac{\gamma_1}{\beta \lambda_2}\textcolor{blue}{.}$$
\end{lemma}
\begin{lemma}\label{lemma21} 
	\cite{Nguyen2021a} Consider a dynamic system
	$$\dot x_i=-\kappa(\sum\limits_{j = 1}^N a_{ij}(x_i-x_j)+b_{ii}(x_i-x_0)).$$
	There exist a positive constant $\gamma_2$, if $\kappa>0$, $x_i$ satisfies  
	$$\mathop {\lim }\limits_{t \to \infty } \left| {x_i-x_0} \right| \leq \frac{\gamma_2\Phi}{\kappa},$$
	where $\Phi=\mathop {max}\limits_{t<\infty}{\vert \dot x_0\vert}$.
\end{lemma}
\par Based on the above analysis, we have the following result about the distributed control \textcolor{blue}{schemes} \eqref{6} and \eqref{7}. And the relevant proof will be given next.
\begin{theorem}
	\label{TH1}
	Consider a multi-BESS system contains $n$ BESS equipped with the dynamics \eqref{5a}-\eqref{5e}. Let the communication topology satisfies Assumptions \ref{assum1}. Then, if $\alpha $, $\kappa $, $\beta  > 0$, Problem 1 will be solved with the distributed control algorithm \eqref{6} and \eqref{7} in asymptotic manner. And we have $\mathop {\lim }\limits_{t \to \infty } \left| {\hat E_{a,i}-E_a} \right| \leq \frac{\gamma}{\beta \lambda_2}$, where $\gamma=\Vert E\Vert_\infty<+\infty$.
\end{theorem}

\begin{proof}
	Define estimation error and frequency modulation error as follow
	$$\Delta P_i=\hat P_{a,i}-P_a,$$ $$\Delta E_i=\hat E_{a,i}-E_a,$$ $$\Delta \omega_i=\omega_i-\omega^r.$$
	Denote some stack vectors $\Delta P = {[\Delta {P_1},\Delta {P_2}, \cdots ,\Delta {P_n}]^T}$, ${\hat P_a} = {[{\hat P_{a,1}},{\hat P_{a,2}}, \cdots ,{\hat P_{a,n}}]^T}$, ${\hat E_a} = {[{\hat E_{a,1}},{\hat E_{a,2}}, \cdots ,{\hat E_{a,n}}]^T}$, $\Delta \omega  = {[\Delta {\omega _1},\Delta {\omega _2}, \cdots ,\Delta {\omega _n}]^T}$, and $\omega  = {[{\omega _1},{\omega _2}, \cdots ,{\omega _n}]^T}$.
	\par Next, the convergence of the average load power estimator will be explained as follow.
	Construct the following candidate Lyapunov function
	$${V_P} = \frac{1}{2}\sum\limits_{i = 1}^n {\Delta P_i^2}.$$
	Derivative of ${V_P}$ with respect to time is formed as
	$$\begin{aligned}
		\dot V_P=&\sum\limits_{i = 1}^n \Delta {P_i} \Delta \dot P_i 
		=-\kappa\sum\limits_{i = 1}^n \Delta P_i \mathop \sum \limits_{j = 1}^n a_{ij}(\hat P_{a,i}- \hat P_{a,j})\\
		=&-\kappa\sum\limits_{i = 1}^n \Delta P_i \mathop \sum \limits_{j = 1}^n a_{ij}(\Delta P_i-\Delta P_j)
		=-\frac{\kappa}{2}\sum\limits_{i,j = 1}^n a_{ij}(\Delta P_i-\Delta P_j)^2< 0\textcolor{blue}{.}
	\end{aligned}$$
	So, $V_P$ converges to 0 asymptotically and $\mathop {\lim }\limits_{t \to \infty } {\Delta {P_i}} = 0$. That is,
	$\mathop {\lim }\limits_{t \to \infty } {\hat P_{i,a}} = {P_a},\forall i \in \{ 1,2, \cdots ,N\}.$
	\par Combined with \eqref{6a} and \eqref{6b}, it can be concluded that
	$$\dot {\hat E}_{a,i}=\dot E_i-\alpha q_i-\beta\sum\limits_{j \in N_i} a_{ij}(\hat E_{a,i}-\hat E_{a,j})-v_i.$$
	With the help of the mapping in Lemma \ref{lemma1} and \eqref{6c}, the following state-space model can be obtained,
	$$\left[ {\begin{array}{*{20}{c}}
			{\dot y}\\
			{\dot w}
	\end{array}} \right] = A\left[ {\begin{array}{*{20}{l}}
			y\\
			w
	\end{array}} \right]\textcolor{blue}{,}$$
	where $y = {\hat E_a} - {E_a}$, $w = v - \bar v$, $\bar v = {\Pi _N}(\dot E + \alpha E)$, $A = \left[ {\begin{array}{*{20}{c}}
			{ - \alpha {I_N} - \beta L}&{ - {I_N}}\\
			{\alpha \beta L}&0
	\end{array}} \right]$.
	Then, it can be inferred in terms of Lemma \ref{lemma1} that
	${y_i}(t) \to  - \frac{{{\alpha ^{ - 1}}}}{N}\mathop \sum \limits_{j = 1}^N {w_j}(0)$, ${w_i}(t) \to \frac{1}{N}\mathop \sum \limits_{j = 1}^N {w_j}(0)$, as $t \to \infty $, $\forall i \in \{ 1, \ldots ,N\} $.
	According to Lemma \ref{lemma2}, $\mathop {\lim }\limits_{t \to \infty } \left| {y_i} \right| \leq \frac{\gamma}{\beta \lambda_2}$, i,e., $\mathop {\lim }\limits_{t \to \infty } \left| {\hat E_{a,i}-E_a} \right| \leq \frac{\gamma}{\beta \lambda_2}$.
	\par Lastly, according to Lemma \ref{lemma21} and \eqref{7}, we have 
	$$\mathop {\lim }\limits_{t \to \infty } \left| {\omega^*_i-\hat P_{a,i}-\omega^r} \right|=\mathop {\lim }\limits_{t \to \infty } \left| {\omega_i-\omega^r} \right| \leq \frac{\gamma_2\Phi}{\kappa},$$
	where $\Phi=\mathop {max}\limits_{t<\infty}{\vert \dot {\hat P}_{a,i}\vert}$. From a different perspective, combined with \eqref{5e}, \eqref{6c} and \eqref{7}, the following derivation is obtained
	$$\begin{aligned}
		\Delta \dot \omega=&\dot \omega={{\dot {\omega }^*}} - {{\dot {\hat P}}_a}
		=- \kappa {\cal L}{\omega ^*} - \kappa {\cal B}({\omega ^*} - \hat P_a - {\omega ^r}) + \kappa {\cal L}{{\hat P}_a}\\
		=&- \kappa {\cal H}(\omega  - {\omega ^r})
		=- \kappa {\cal H}\Delta\omega\textcolor{blue}{.}
	\end{aligned}$$
	Then, $\Delta\omega$ can be obtained by integrating $\Delta \dot \omega $ with respect to time
	$\Delta \omega (t) = {e^{ - \kappa {\cal H}t}}(\omega(0)-\omega^r)$,
	where $\omega (0)$ is the synthesis vector of the initial value of each BESS's frequency. It can be deduced that the eigenvalues of ${\cal H}$ are non-zero with negative real part. Accordingly, $\mathop {\lim }\limits_{t \to \infty } \Delta \omega (t) = 0$, i.e., $\mathop {\lim }\limits_{t \to \infty } {\omega _i} = {\omega ^r},\forall i \in \{ 1,2, \cdots ,N\}.$
	Then, Problem 1 has been solved in asymptotic manner.
	$\hfill\blacksquare $
\end{proof}
\par So far, we complete the proof of Theorem 1.
\begin{remark}
	\label{rem1}
	The dynamic and steady-state performance of $\hat E_{a,i}$ are influenced by different parameters, respectively. Dynamic tracking performance is closely positively correlated with $\alpha$, that is, as $\alpha$ increases, the tracking rate becomes faster. The consensus steady-state performance, according to Lemma \ref{lemma2}, can be improved by selecting \textcolor{blue}{a} larger $\beta$ under a given communication graph. The consensus dynamic performance depends on $\alpha$ and $\beta$.
\end{remark}
\begin{remark}
	This solution is completely distributed, as each agent utilizes information from itself and its neighbors. 	In addition, the derivative of SoC, i.e. proportional power, is not required in real-time. 
	But this does not mean that there is no need to measure power. Occasionally, when switching between constant
	loads, the estimated power sharing state needs to be reset to the power at the time of load switching. 
	So, the complete dynamic of the system is a hybrid integrator with load-driven jumps, that is,
	$$\left\{\begin{array}{l}
		{\hat P}_{a,i}(t^+)=P_i,\, \mathrm{if}\,D\,\mathrm{changs\,at}\,t,\\
		\dot{\hat P}_{a,i}=-\kappa\sum\limits_{j \in N_i} a_{ij}({\hat P}_{a,i}-{\hat P}_{a,j}),\,\mathrm{otherwise}.
	\end{array}\right.$$
	The estimator below can serve as an alternative to this solution,
	$$\dot{\hat P}_{a,i}=-\kappa(\sum\limits_{j \in N_i} a_{ij}({\hat P}_{a,i}-{\hat P}_{a,j})+b_{ii}(\hat P_{a,i}-P_a)).$$
	At this point, simply update $P_a$ at the time of load switching.
\end{remark}
\subsection{Distributed finite-time control scheme}
\par In this section, we present the distributed finite-time scheme to solve Problem 1. In order to illustrate the subsequent design and related proof, we give the following lemmas.
\newtheorem{Definition}{Definition}
\begin{lemma}\label{lemma4}
	\cite{Nguyen2021a} If ${y_1},{y_2}, \cdots ,{y_n} \ge 0$ and $0 < r < p$, then
	$${(\sum\limits_{i = 1}^n {y_i^p} )^{\frac{1}{p}}} \le {(\sum\limits_{i = 1}^n {y_i^r} )^{\frac{1}{r}}}$$.
\end{lemma}
\begin{lemma}\label{lemma5}
	\cite{Zhang2012} Denote $x$ a $n$-dimensional real column vector. Then we have $x^T{\cal L}x \ge {\lambda_2}{x^T}x$, ${x^T}{\cal H}x \ge {\lambda_m}{x^T}x$, and $\Vert x\Vert_1\ge\Vert x\Vert_2$.
\end{lemma}
\begin{lemma}\label{lemma6}
	\cite{Forti2006} Consider the following system dynamic as $\dot x = u$. If the Lyapunov function $V$ with respect to the state $x$ satisfies $\dot V \le  - K{V^\alpha }$, where $K > 0$ and $0 < \alpha  < 1$. Then $V$ will converges to 0 in finite time $T=\frac{V^{1-\alpha}(0)}{K(1-\alpha)}$. At the same time, the state $x$ reaches stable.	
\end{lemma}
\par \textcolor{blue}{A} finite time estimator of SoC balance of each battery can be designed as \eqref{8a} and \eqref{8b}
\begin{subequations}
	\label{8}
	\begin{equation}
		\label{8a}
		{\hat E_{a,i}} = \sum\limits_{j=1}^n {{a_{ij}}({q_i} - {q_j})}  + {E_i},
	\end{equation}
	\begin{equation}
		\label{8b}
		\dot q_i=-\alpha sign(\sum\limits_{j=1}^n {{a_{ij}}({{\hat E}_{a,i}} - {{\hat E}_{a,j}})} ),
	\end{equation}
\end{subequations}
where ${\hat E_{a,i}}$ is the estimation of the average of SoC of battery $i$, ${q_i}$ is \textcolor{blue}{a} internal state, $\alpha  > 0$ is a designed parameter, $sign( \bullet )$ represents the symbolic function. In addition, the power sharing estimator and the nominal frequency controller are designed as \eqref{9a} and \eqref{9b},
\begin{subequations}
	\label{9}
	\begin{equation}
		\label{9a}
		\dot {\hat P}_{a,i}=-\beta_1 \mathop\sum\limits_{j=1}^n a_{ij}sign(\hat P_{a,i}-\hat P_{a,j})\vert\hat P_{a,i}-\hat P_{a,j}\vert^\eta,
	\end{equation}
	\begin{equation}
		\label{9b}
		\dot \omega_i^*=-\beta_2 sign(\sum\limits_{j=1}^n a_{ij} (\omega_i^*-\omega_j^*)+b_{ii}(\omega_i^*-\hat P_{a,i}-\omega ^r)).
	\end{equation}
\end{subequations}
\par With these above lemmas in hand, the following result can be easily obtained based on the designed estimators.
\begin{theorem}
	Consider a power network with $n$ BESSs subjecting to the dynamics \eqref{5}. Let the communication topology satisfy Assumption \ref{assum1}. Then, if $\alpha>\frac{{\sqrt n {P_\Sigma }}}{{3600{\lambda _2}}}$, $\beta_1>0$, $\beta_2>\phi$ and $0\le\eta< 1$, where $\phi=\Vert \dot{\hat P}(0)\Vert_\infty$, Problem 1 will be solved with the distributed control algorithm \eqref{8} and \eqref{9} in finite time. That is, active power sharing and SoC balance states can be estimated at $t=T_P$ and $t=T_E$ respectively. Besides, frequency restores to $\omega^r$ at $t =T_\omega $, where the nominal frequency of each BESS converge to  $({\omega^r+P_a})$ at $t=T_\omega$.
\end{theorem}
\begin{proof}
	Firstly, the convergence of the average load power estimator is explained under the condition of $0 \le \eta  < 1$ and $\beta  > 0$. Select the following Lyapunov function
	$$V_P=\Delta P^T\Delta P.$$
	Find its derivative as
	$$\begin{aligned}
		\dot V_P=& 2\sum\limits_{i = 1}^N \Delta P_i\Delta \dot P_i
		=2\sum\limits_{i = 1}^N \Delta P_i (-\beta_1 \mathop\sum\limits_{j=1}^n a_{ij}sign(\hat P_{a,i}-\hat P_{a,j})\vert\hat P_{a,i}-\hat P_{a,j}\vert^\eta) \\
		=&-\beta_1\sum\limits_{i,j = 1}^N (\Delta P_i-\Delta P_j)sign(\Delta P_i-\Delta P_j)\vert\Delta P_i-\Delta P_j\vert^\eta \\
		=&-\beta_1\sum\limits_{i,j = 1}^N a_{ij}\vert \Delta P_i-\Delta P_j\vert^{\eta+1}.
	\end{aligned}$$
	According to Lemma \ref{lemma4} and Lemma \ref{lemma5}, $\dot V$ follows the following derivation
	$$\begin{aligned}
		\dot V_P=&-\beta_1\sum\limits_{i,j = 1}^N a_{ij}\vert \Delta P_i-\Delta P_j\vert^{\eta+ 1}
		\le -\beta_1(\sum\limits_{i,j = 1}^n a_{ij}^\frac{2}{\eta + 1}\vert\Delta P_i-\Delta P_j\vert^2)^\frac{1}{2}\\
		=&-\beta_1[2\Delta P^T{\cal L}\Delta P]^\frac{\eta+1}{2}
		\le-\beta_1(2\lambda_2)^\frac{\eta+ 1}{2}[V_P]^\frac{\eta+ 1}{2}.
	\end{aligned}$$
	By Lemma \ref{lemma6}, $\mathop {\lim }\limits_{t \to {T_P}} {V_P} = 0$, where 
	\begin{equation}
		\label{10}
		T_P=\frac{\Vert \Delta P(0)\Vert^{1-\eta}}{(1-\eta)\beta_1(2\lambda_2)^\frac{\eta+1}{2}}=\frac{(\Vert \Delta P(0)\Vert\sqrt{2\lambda_2})^{1-\eta}}{2\lambda_2\beta_1(1-\eta)}.
	\end{equation}
	That is $\mathop {\lim }\limits_{t \to {T_P}} \left| {{{\hat P}_{a,i}} - {P_a}} \right| = 0$ and ${\hat P_{a,i}} = {P_a}$ when $t \ge {T_P}$ for $\forall i \in \{ 1,2, \cdots ,N\} $. 
	\par Rewrite \eqref{9a} in a matrix form as
	$$\dot\omega^*=-\beta_2 sign({\cal H}(\omega^*-\hat P_a-\omega ^r)).$$
	Define a new variable $\delta\omega={\cal H}(\omega^*-\hat P_a-\omega ^r)$. Select a candidate Lyapunov function as
	$$V_\omega=\frac{1}{2}\delta\omega^T {\cal H}^{-1}\delta\omega,$$
	whose derivative is, according to Lemma \ref{lemma5}
	\begin{equation}
		\label{11}
		\begin{aligned}
			\dot V_\omega=&\delta\omega^T{\cal H}^{-1}{\cal H}(\dot\omega^*-\dot{\hat P}_a)
			\le \delta\omega^T\dot\omega^*-\delta\omega^T\dot{\hat P}_a
			=-\beta_2\delta\omega^Tsign(\delta\omega)-\delta\omega^T\dot{\hat P}_a\\
			=&-\beta_2\sum\limits_{i = 1}^n\vert\delta\omega_i\vert-\delta\omega^T\dot{\hat P}_a
			\le -\beta_2\Vert\delta\omega\Vert_1+\phi\Vert\delta\omega\Vert_1\\
			\le&-(\beta_2-\phi)(2\lambda_M)^{\frac{1}{2}} V_\omega^{\frac{1}{2}},
		\end{aligned}
	\end{equation}
	which implies, according to Lemma \ref{lemma6}, if $\beta_2>\phi$, the nominal frequency $\omega _i^*$ converges to $(\hat P_{a,i}+\omega ^r)$ in finite time 
	\begin{equation}
		\label{12}
		T_\omega=\frac{\sqrt {2V_\omega(0)}}{(\beta_2-\phi)\sqrt{\lambda_M}}
	\end{equation}
	for $\forall i \in \{ 1,2, \cdots ,N\} $. Thus, under the action of droop controller, ${\omega _i} = {\omega ^r}$  when $t\ge T_\omega$.
	\par Next, the finite time convergence of the SoC balance estimator will be proved.
	\par By constructing a matrix form of \eqref{8a} and \eqref{8b}, one can get
	$${\dot {\hat E}_a} = {\cal L}{u^q} + \dot E,$$
	$$\dot q = {u^q},$$
	$${u^q} =  - \alpha sign(\mu),$$
	$$\mu  = {\cal L}{\hat E_a}.$$
	As for the estimation error $\Delta E$ of SoC, the following two equations hold true,
	$$\begin{array}{l}
		\Delta \dot E = {{\dot {\hat E}}_a} - \frac{1}{N}{1_N}1_N^T\dot E
		= {\cal L}{u^q} + ({I_N} - \frac{1}{N}{1_N}1_N^T)\dot E
	\end{array},$$
	$${\cal L}\Delta E = {\cal L}({\hat E_a} - {E_a}) = {\cal L}{\hat E_a}.$$
	Denote ${\cal L}({\cal L}^+) = {I_N} - \frac{1}{N}{1_N}1_N^T$, where ${\cal L}^+$ is the generalized inverse of ${\cal L}$. Then, we arrive at
	$$({\cal L}^+)\mu  = ({\cal L}^+){\cal L}{\hat E_a} = {\hat E_a}.$$
	Select the Lyapunov function with respect to $\Delta E$ as
	$$V = \frac{1}{2}\Delta {E^T}\Delta E.$$
	Because $V$ is a well-defined continuous function with continuous first derivative function, the first derivative function can be derived as
	$$\begin{array}{l}
		\dot V = \Delta {E^T}\Delta \dot E
		= \Delta {E^T}({{\dot {\hat E}_a}} - {{\dot E}_a})
		= \Delta {E^T}({\cal L}{u^q} +{\cal L}({\cal L}^+)\dot E)
		= {\mu ^T}{u^q} + {\mu ^T}({\cal L}^+)\dot E.
	\end{array}$$
	The second term on the right side of the above equation satisfies the following inequality,
	$${\mu ^T}({\cal L}^+)\dot E \le {\sigma _{\max }}({\cal L}^+)\sqrt n \parallel \mu {\parallel _1}\parallel \dot E{\parallel _\infty },$$
	where ${\sigma _{\max }}({\cal L}^+)$ denotes the maximum singular value and ${\sigma _{\max }}({\cal L}^+) = \frac{1}{{{\lambda _2}}}$. Also,
	$${\mu ^T}{u^q} =  - \alpha  \mathop \sum \limits_{i = 1}^n\left| {{\mu _i}} \right| =  - \alpha \parallel {\mu}{\parallel _1}.$$
	As a result, $\dot V$ satisfies the following derivation
	$$\begin{aligned}
		\dot V \le& \frac{1}{{{\lambda _2}}}\sqrt n \parallel \mu {\parallel _1}\parallel \dot E{\parallel _\infty } - \alpha \parallel {\mu}{\parallel _1}
		= (\frac{1}{{{\lambda _2}}}\sqrt n {\left| {\dot E} \right|_{\max }} - \alpha )\parallel {\mu}{\parallel _1}\\
		\le& (\frac{1}{{{\lambda _2}}}\sqrt n \frac{1}{{3600}}{\left| P \right|_{\max }} - \alpha )\parallel {\mu}{\parallel _1}.
	\end{aligned}$$
	In general, ${P_i} \le {P_\Sigma }$. If $\alpha  > \frac{{\sqrt n {P_\Sigma }}}{{3600{\lambda_2}}}$, according to Lemma \ref{lemma5}, we arrive at
	\begin{equation}
		\label{13}
		\begin{aligned}
			\dot V \le& (\frac{\sqrt {n} P_\Sigma}{3600\lambda_2}-\alpha)\Vert\mu\Vert_1
			\le (\frac{\sqrt {n} P_\Sigma}{3600\lambda_2}-\alpha)\Vert\mu\Vert_2
			=(\frac{\sqrt {n} P_\Sigma}{3600\lambda_2}-\alpha)\sqrt{\Delta {E^T}{L^2}\Delta E}\\
			\le &(\frac{\sqrt {n} P_\Sigma}{3600\lambda_2}-\alpha){\lambda_2}{V^{\frac{1}{2}}}.
		\end{aligned}
	\end{equation}
	According to Lemma \ref{lemma6}, SoC balance error converges to 0 within a finite time denoted by ${T_E}$ and
	\begin{equation}
		\label{14}
		T_E=\frac{2\Vert\Delta E(0)\Vert}{\alpha\lambda_2-\frac{\sqrt {n} P_\Sigma}{3600}}.
	\end{equation}
	\par So far, we have completed the proof of Theorem 2. $\hfill\blacksquare$
\end{proof}
\begin{remark}
	\label{rem3}
	From \eqref{10}, it can be seen that the convergence rate of $\Delta P$ can be accelerated by selecting a larger parameter $\beta_1$ \textcolor{blue}{(or $\beta_2$)} and a communication topology with a larger minimum non-zero eigenvalue, respectively. The role of parameter $\eta$ is not yet clear. This is because given different system initial values, communication topology, and parameters \textcolor{blue}{$\beta_1$, $\beta_2$}, the influence of parameter $\eta$ on convergence time is different. Therefore, this requires further analysis based on different situations, and inevitably requires the use of trial and error methods.
\end{remark}
\begin{remark}
	\label{rem4}
	From \eqref{12}, it can be seen that the convergence time of frequency decreases with the increase of the maximum eigenvalues and parameters of the communication topology, respectively. However, $\beta_2$ no longer affects the convergence time by selecting $\beta_2\ge\phi+1$. That is, \eqref{11} can be further deduced as $\dot V_\omega\le-(2\lambda_M)^\frac{1}{2}V_\omega^\frac{1}{2}$. Then, according to Lemma \ref{lemma6}, $T_\omega=\frac{\sqrt {2V_\omega(0)}}{\sqrt{\lambda_M}}$. Similarly, \eqref{14} shows that the convergence time decreases with the increase of $\alpha$ and $\lambda_2$. When $\alpha\lambda_2\ge\frac{\sqrt {n} P_\Sigma}{3600}+1$ (or $\alpha\ge\frac{\sqrt {n} P_\Sigma}{3600\lambda_2}+1$) is selected, \eqref{13} can be further deduced as $\dot V\le -V^\frac{1}{2}$ (or $\dot V\le -\lambda_2 V^\frac{1}{2}$), which leads to $T_E=2\Vert\Delta E(0)\Vert$ (or $T_E=\frac{2\Vert\Delta E(0)\Vert}{\lambda_2}$). This implies that by selecting $\alpha\lambda_2\ge\frac{\sqrt {n} P_\Sigma}{3600}+1$ or $\alpha\ge\frac{\sqrt {n} P_\Sigma}{3600\lambda_2}+1$, the convergence time $T_E$ is not affected by $\alpha$ and $\lambda_2$ or only by $\lambda_2$.
\end{remark}
\section{Some simulation cases}
\par In this section, several cases are designed to verify the two proposed schemes. The \textcolor{blue}{modified} IEEE-57 bus is selected as the load distribution exhibited in Figure \ref{fig2}\textcolor{yellow}{, where 7 BESSs are set on buses 1, 2, 3, 6, 8, and 9}. The communication topologies are illustrated in Figure \ref{fig3}. $Matpower$, which is used to calculate power flow in the simulation, is a $Matlab$ toolkit. $MatDyn$ is \textcolor{blue}{selected} as the Simulation toolkit based on $Matpower$. Here are six case studies designed to verify and test the performance of the two designed solutions.
\par These two schemes, under a fixed graph, are test in Case 1 and Case 4, respectively.
\par To test the effectiveness of two schemes under time-varying communication topologies, a switching sequence with an interval of 20 minutes, i.e., (a) $\to$ (b) $\to$ (c) $\to$ (d) in Figure \ref{fig3}, is executed in Case 2 and Case 5.
\par In Case 3 and Case 6, the influence of parameters on the results are studied trial and error.
\begin{figure}
	\centering
	\subfigure[The \textcolor{blue}{modified} IEEE-57 bus system\label{fig2}] {\includegraphics[width=7.5cm]{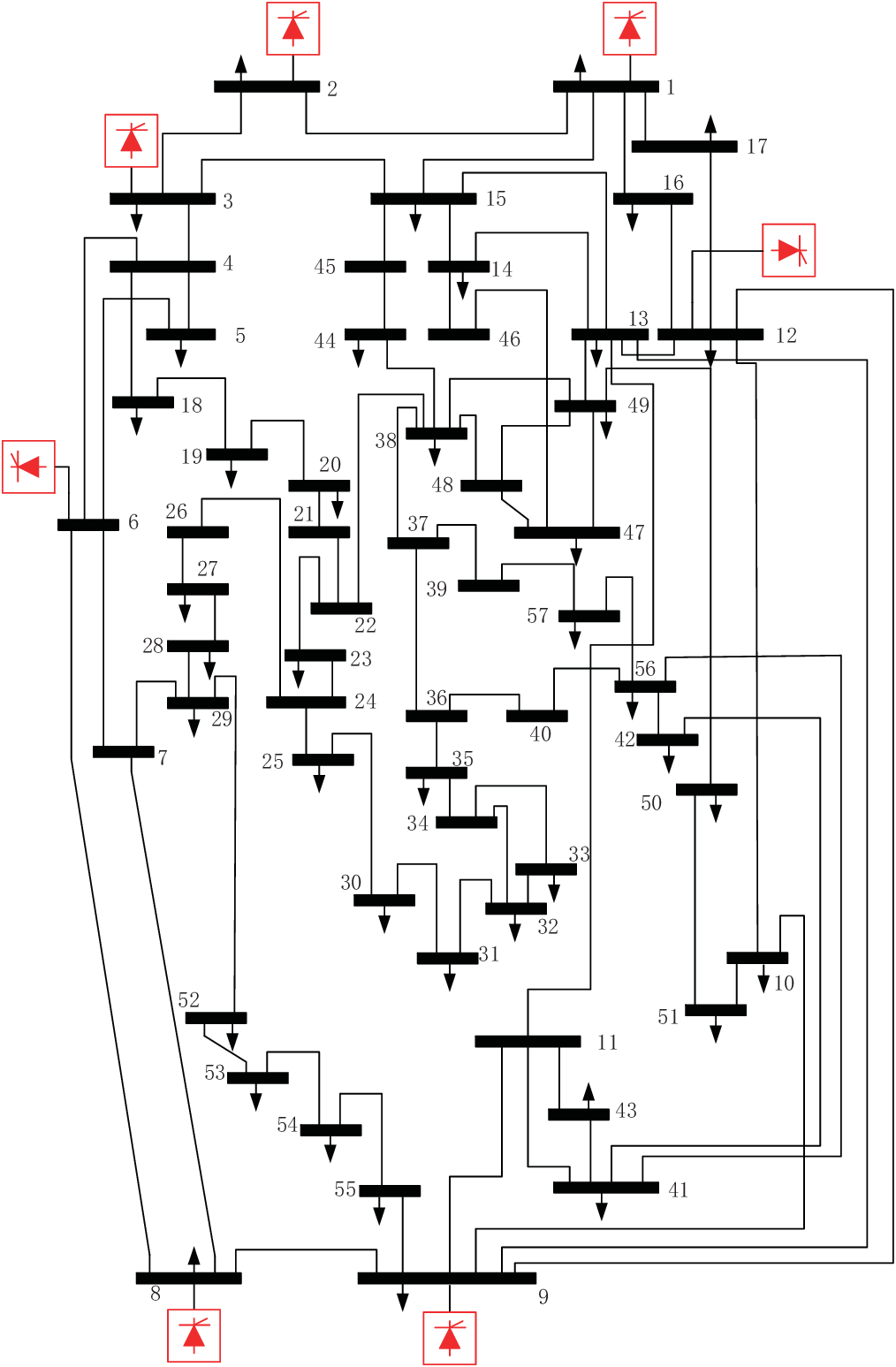}}\\
	\subfigure[The communication topology graph\label{fig3}] {\includegraphics[width=8.0cm]{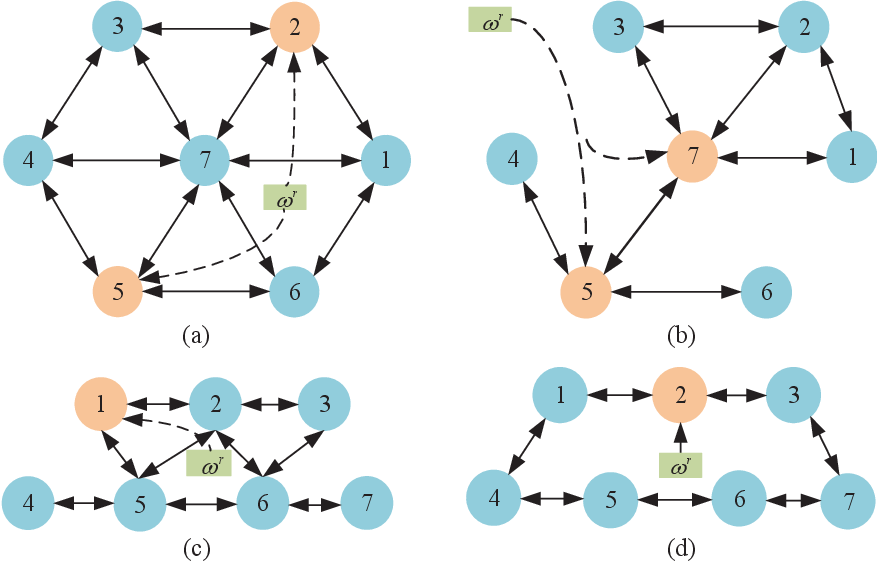}}\\
	\caption{The \textcolor{blue}{modified} IEEE-57 bus system and its communication topology}
	\label{figs}
\end{figure}
\subsection{Case 1. The Test on the Asymptotic Algorithm}
\par In order to verify the effectiveness of the strategy proposed in asymptotic manner, some operations are implemented during the simulation. The whole system works in discharging mode. The secondary controllers based on droop control work at $t = 10$. When $t = 40$, there is a total load decrease of $1.26\;p.u.$, which is evenly distributed at each load node. Besides, \textcolor{yellow}{$0.06\;p.u.$} increase at each load node occurs at $t = 75$.
\par Figure \ref{fig4} is the simulation result. Figure \ref{fig4} reveals that the estimation of SoC balance, active power sharing are reached, and the frequency restores to the reference. In addition, the SoC balance can be well maintained regardless of how the load changes and active power of each BESS outputs stably.
\par To illustrate the progressiveness of the scheme we designed, two existing distributed schemes are compared, one of which is fully driven \cite{8076899} and the other is under driven \cite{Nguyen2021a}. The simulation results are shown in Figure \ref{fig41}. Comparing Figure \ref{fig4}, it can be seen that previous schemes have performed poorly in addressing the secondary control and SoC balance of BESSs with capacity degradation. Especially in the underactuated scheme, it makes it difficult to achieve precise consensus on the SoC, proportional power, and frequency of each BESS, respectively. In contrast, our designed solution improves the traditional droop control approach of distributing power according to rated capacity, so that power is sharing according to the current SoC level. This enables all the SoCs of BESS to converge to 0 simultaneously, while ensuring accurate frequency restoration.
\begin{figure}
	\centering
	\includegraphics[width=8.0cm]  {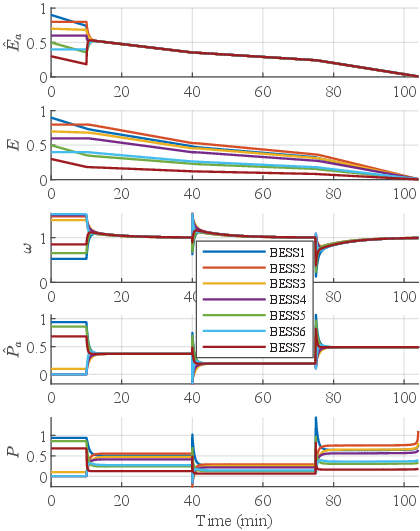} \caption{The simulation result of the asymptotic algorithm in Case 1.\label{fig4}}
\end{figure}
\begin{figure}
	\centering
	\subfigure[Simulation results under an underactuated scheme in \cite{Nguyen2021a}. \label{fig5c}] {\includegraphics[width=8.0cm]{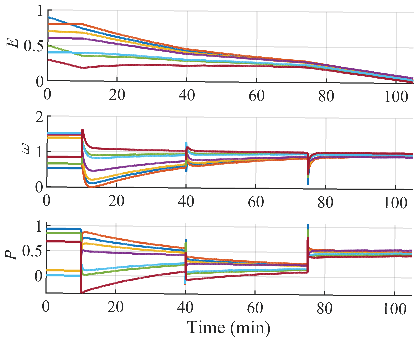}}\\
	\subfigure[Simulation results under a full drive scheme in \cite{8076899}.\label{fig5cc}] {\includegraphics[width=8.0cm]{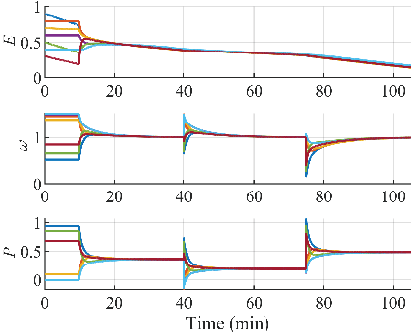}}\\
	\caption{Simulation results under previous schemes.}
	\label{fig41}
\end{figure}
\subsection{Case 2. The Test on the Asymptotic Algorithm with Time-Varying Communication Topologies}
\par In this case, the performance of the asymptotic scheme in dealing with switching topologies is studied. After the controller is activated at $t=10$, a total load decrease of $1.26\;p.u.$ and increase of $2.52\;p.u.$ occur at $t=35$ and $t=55$, respectively. According to the previously predetermined topology switching frequency and order, the relevant simulation results are shown in Figure \ref{fig5}.
\par From this result, it can be seen that each variable can still converge under time-varying topologies, but with slightly different rates. This is related to the current communication topology. This indicates the robustness of the designed scheme to switching topologies. In addition, as long as connectivity is maintained, even if some communication links fail, this algorithm remains effective, ensuring the stability and security of BESS.
\begin{figure}
	\centering
	\includegraphics[width=8.0cm]  {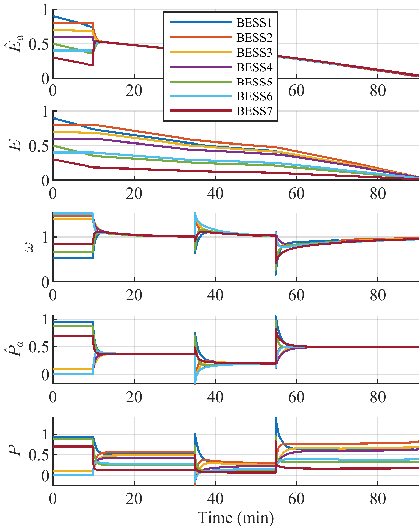} \caption{The simulation result of the asymptotic algorithm with time-varying communication topologies in Case 2.\label{fig5}}
\end{figure}
\subsection{Case 3. Influence of Parameters on the Asymptotic Algorithm}
\par Different parameters value groups of $\alpha$, $\beta$, and $\kappa$ in the distributed asymptotic controller are set in this case to grasp the influence of each parameter on the results and verify the correctness of the discussion in Remark \ref{rem1}. The simulation results of the real time SoC balance estimation, power sharing state estimation and frequency restoration are shown in Figure \ref{fig6} and \ref{fig7}.
\begin{figure}
	\centering
	\subfigure[Effect of different $\alpha$ and $\beta$ on $\hat E_{a,i}$ under the asymptotic algorithm in Case 3.\label{fig6}] {\includegraphics[width=8.0cm]{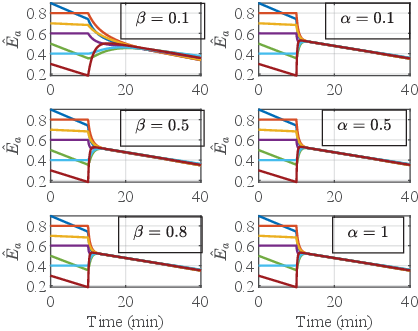}}\\
	\subfigure[Effect of different $\kappa$ on $\omega_i$ and $\hat P_{a,i}$ under the asymptotic algorithm in Case 3.\label{fig7}] {\includegraphics[width=8.0cm]{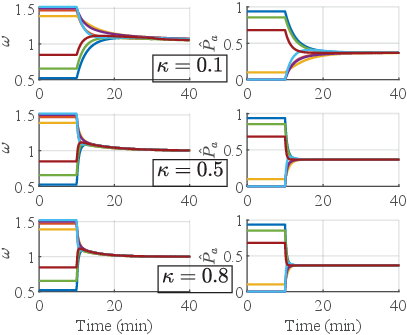}}\\
	\caption{The influence of parameters of asymptotic algorithms on simulation results.}
	\label{figc3}
\end{figure}
\par From the simulation results in Figure \ref{fig6}, different values of $\alpha$ have no obvious effect on the results. However, the real time SoC balance estimation is mainly affected by $\beta$. Specifically, slow speed of SoC balance estimation become more serious with a smaller $\beta$. For power sharing estimation and frequency restoration, selecting a larger $\kappa$ can accelerate the convergence rate. All of these indicate the correctness of the discussion in Remark \ref{rem1}.
\subsection{Case 4. The Test on the Finite Time Control Algorithm}
\par To test the effect of the proposed finite time control scheme, the same events and load changes as Case 1 are implemented in this case. Figure \ref{fig8} are the simulation results in the discharging mode. 
\par Compared with Case 1, although the estimation of the real time SoC balance and power sharing can be well estimated as described in Figure \ref{fig8}. Besides, frequency can be restored to the reference value. SoCs can ultimately reach 0 simultaneously and BESSs have stable power output. At this time, the high-frequency chatting behavior is not obvious.
\begin{figure}
	\centering
	\includegraphics[width=8cm]  {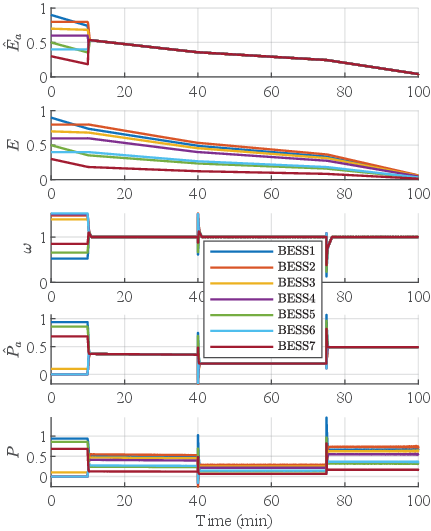} \caption{The simulation result of the finite time algorithm in Case 4.\label{fig8}}
\end{figure} 
\subsection{Case 5. The Test on the Finite-time Algorithm with Time-Varying Communication Topologies}
\par In this case, the performance of the finite-time scheme in dealing with switching topologies is studied. After the controller is activated at $t=10$, a total load decrease of $1.26\;p.u.$ and increase of $2.52\;p.u.$ occur at $t=35$ and $t=55$, respectively. According to the previously predetermined topology switching frequency and order, the relevant simulation results are shown in Figure \ref{fig9}.
\par From this result, it can be seen that each variable can still converge under time-varying topologies, and with no significantly changed rate. This indicates, similar to the conclusion in Case 3, this algorithm exhibits robustness to time-varying topologies and solves the problems proposed in this article while ensuring the connectivity of the communication topology.
\begin{figure}
	\centering
	\includegraphics[width=8cm]  {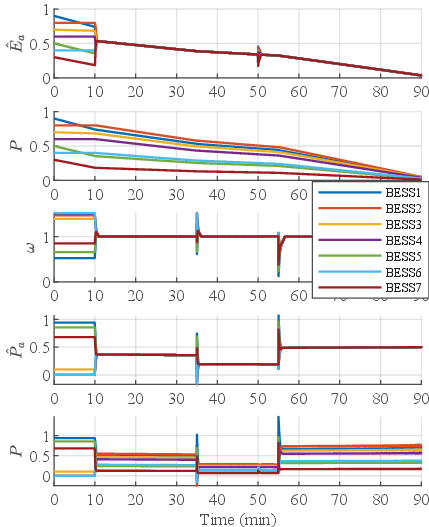} \caption{The simulation result of the finite-time algorithm with time-varying communication topologies in Case 5.\label{fig9}}
\end{figure}
\subsection{Case 6. Influence of Parameters on the Finite-Time Algorithm} 
\par To explore the impact of relevant parameters $\alpha$, $\beta$, and $\eta$ on the system behavior and validate the discussion in Remarks \ref{rem3} and \ref{rem4}, simulations are conducted using different values of each parameter based on Case 1. Corresponding to Figure \ref{fig3}, $\lambda_2=2$ and $\lambda_m=0.2412$. Besides, $\Vert \Delta P(0)\Vert=1.0671$, $P_{\Sigma}=2.575$, and $\phi=3.057\beta_1$. Hence, through a simple verification of \eqref{10} with $0\le\eta<1$, $T_P$ is negatively correlated with $\eta$. In addition, if selecting $\alpha\ge 1.0557$, according to Remark \ref{rem4}, $T_E$ remains unchanged. And according to Remark \ref{rem4}, $T_\omega$ remains unchanged by selecting $\beta_2\ge 3.057\beta_1+1$. The simulation results are shown in Figure \ref{fig10}-\ref{fig12}.
\par It can be deduced from Figure \ref{fig10} that with a larger $\beta_1$, consensus on $\hat P_{a,i}$ is accelerated. However, the conclusion about $\eta$ is exactly the opposite. From the simulation results of $\hat E_{a,i}$, it can be seen that the increase of $\alpha$ accelerates consensus rate, which is, however by selecting $\alpha>1.0557$, no longer changes. On the contrary, the shaking phenomenon becomes more pronounced. For $\beta_2$, from Figure \ref{fig12}, a larger value can indeed accelerate frequency recovery, but after exceeding a certain value, there is no significant change in the convergence rate. Also, it is necessary to select the proper $\alpha$ in order to obtain the high quality system performance and ideal consensus speed. The above phenomena all verify the correctness of the discussion in Remark \ref{rem3} and \ref{rem4}.
\begin{figure}
	\centering
	\subfigure[Effect of different $\beta_1$ and $\eta$ on $\hat P_{a,i}$ under the finite-time algorithm in Case 6.\label{fig10}] {\includegraphics[width=8.0cm]{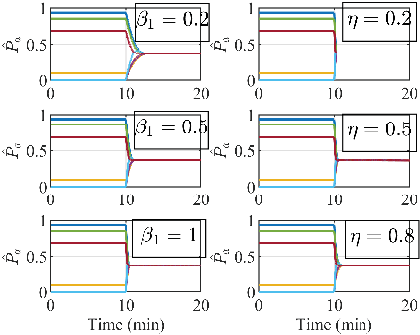}}\\
	\subfigure[Effect of different $\alpha$ on $\hat E_{a,i}$ under the finite-time algorithm in Case 6.\label{fig11}] {\includegraphics[width=8.0cm]{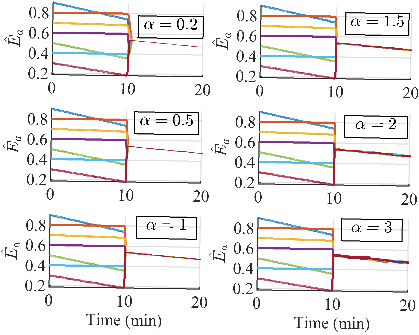}}\\
	\caption{The influence of $\beta_1$ and $\alpha$ of the finite-time algorithm on simulation results.}
	\label{figc6}
\end{figure}
\begin{figure}
	\centering
	\includegraphics[width=8cm]  {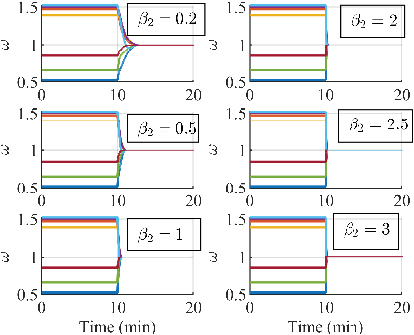} \caption{The influence of $\beta_2$ of the finite-time algorithm on simulation results.\label{fig12}}
\end{figure}
\section{Conclusions}
\par In this article, we propose a novel droop controller with a dynamic droop gain for a BESS to allocate power, and design a distributed asymptotic and finite time secondary controller for it. These efforts are made to address the issue of secondary control and SoC balance considering battery cell capacity degradation. As a result, our designed scheme promotes the SoCs of all BESSs to converge to 0 simultaneously and ensures frequency restore to the reference value through reasonable power allocation. This scheme ensures all BESSs to work in the
discharging mode before depletion, thus avoiding some BESSs working modes from being changed before the entire network is depleted of energy. We do not advocate for consistent SoCs for all BESSs, as this may cause some batteries to enter the next cycle ahead of schedule, leading to modeling errors. Two distributed schemes are designed, namely asymptotic and finite time schemes. The simulation results show that the power is sharing according to the SoC ratio, and the frequency can be restored to the reference value, while the SoCs of all BESSs simultaneously drops to 0. In addition, based on this work, new explorations can be explored in the future, such as intermittent communication and event triggering schemes.
\section{Declaration of conflicting interests}
\par The author(s) declared no potential conflicts of interest with respect to the research, authorship, and/or publication of this
article.
\section{Funding}
\par The author(s) disclosed receipt of the following financial support for the research, authorship, and/or publication of this article: This work was supported by the National Natural Science Foundation of China (Grant No. 62103203), and the General Terminal IC Interdisciplinary Science Center of Nankai University
\section{ORCID iDs}
Yalin Zhang\orcidA{} https://orcid.org/0000-0002-3788-700X\\
Zhongxin Liu\orcidB{} https://orcid.org/0000-0002-3565-4800\\
Fuyong Wang\orcidC{} https://orcid.org/0000-0002-2747-9635\\
Zengqiang Chen\orcidD{} https://orcid.org/0000-0002-1415-4073
\bibliographystyle{SageH}
\bibliography{re.bib}
\end{document}